\documentclass[a4paper]{article}
\usepackage{graphicx}
\usepackage{amsmath,amssymb,amsthm}
\usepackage{float}
\usepackage[plainpages=false,pdfpagelabels,hyperfootnotes=false]{hyperref}
\usepackage[section]{placeins}
\usepackage{booktabs}
\usepackage{color}

\theoremstyle{theorem}

\newtheorem{lemma}{Lemma}

\newtheorem{corollary}{Corollary}
\newtheorem{proposition}{Proposition}

\theoremstyle{definition}

\newtheorem{example}{Example}
\newtheorem{remark}{Remark}

\newcommand{\Keywords}[1]{\par\noindent{\small{Keywords\/}: #1}}
\newcommand{\AMS}[1]{\par\noindent{\small{AMS Classification\/}: #1}}

\floatstyle{ruled}
\floatname{algorithm}{Algorithm}
\newfloat{algorithm}{htb}{alg}

\title{Scheduling arc shut downs in a network to maximize flow over time with a bounded number of
  jobs per time period\thanks{This research was supported by the ARC Linkage Grants no. LP0990739
    and LP1102000524 and HVCCC P/L.}}

\author{Natashia Boland \qquad Thomas Kalinowski \qquad Simranjit Kaur\\
\smallskip
 {\it University of Newcastle, Australia.}
\medskip
}

\begin{document}
\maketitle 

\begin{abstract}
We study the problem of scheduling maintenance on arcs of a capacitated network so as to maximize the total flow from a source node to a sink node over a set of time periods. Maintenance on an arc shuts down the arc for the duration of the period in which its maintenance is scheduled, making its capacity zero for that period. A set of arcs is designated to have maintenance during the planning period, which will require each to be shut down for exactly one time period. In general this problem is known to be NP-hard, and several special instance classes have been studied. Here we propose an additional constraint which limits the number of maintenance jobs per time period, and we study the impact of this on the complexity.

\medskip

\AMS {90C10, 90B10, 68Q25}
\end{abstract}

\section*{Introduction}

We consider the problem of scheduling maintenance jobs on the arcs of a flow network with the
objective of maximizing the throughput over a given time horizon. This problem combines the diverse fields
of scheduling (see for instance~\cite{pinedo2008scheduling}) and network flow optimization, in
particular dynamic network flows, which have been the subject of intense study in recent years;
see, for example,~\cite{koch2011flowsovertime,kotnyek2003annotated,skutella2009introduction}. 

The combination of scheduling and network optimization represents a natural extension to
existing network models, and admits many interesting variants. For example, Tawamalarmi and
Li~\cite{Tawa}, motivated by a problem in highway maintenance, consider a multicommodity flow variant, 
providing complexity results, combinatorial algorithms, and integer programming models. Network optimization
problems and scheduling have also been combined in the context of restoring infrastructure networks
after major disruptions~\cite{nurre2013thesis,nurre2014integrated,nurre2012restoring}) and in network
design over time~\cite{Baxter.etal_Incremental_2014,Kalinowski.etal_Incremental_2015}.

The optimization problem studied in the present paper was originally motivated by annual
maintenance planning for a coal export supply chain~\cite{boland2011optimizing_hvcc}, in which
maximizing the annual throughput is a key concern (see~\cite{liden2014survey} for a comprehensive survey of
mathematical models in railway maintenance scheduling). Boland~\emph{et
  al.}~\cite{boland2012mixed,boland2012scheduling} introduced a general network optimization problem
in which arc maintenance jobs need to be scheduled so as to maximize the total flow in the network
over time. A simplified version of the problem in which all jobs have unit processing time was
studied in~\cite{boland2013unit_time}, and the complexity was determined taking into account certain
instance characteristics, such as special network structures and restrictions on the set of jobs. 

In
the present paper we extend this model by adding the constraint that the number of jobs scheduled in
any time period is bounded by a number $K$ which is given as part of the input. The problem is
defined over a network $N=(V,A,s,t,u)$ with node set $V$, arc set $A$ (we admit parallel arcs having
the same start and end nodes), source $s\in V$, sink $t\in V$ and nonnegative integral capacity
vector $u=(u_a)_{a\in A}$. By $\delta^-(v)$ and $\delta^+(v)$ we denote the set of incoming and
outgoing arcs of node $v$, respectively. We consider this network over a set of $T$ time periods
indexed by the set $[T]:=\{1,2,\ldots,T\}$, and our objective is to maximize the total flow from $s$
to $t$. We are also given a subset $J\subseteq A$ of arcs that have to be shut down for exactly one
time period in the time horizon. In other words, there is a set of maintenance jobs, one for each
arc in $J$, each with unit processing time. In addition, there is a parameter $K$ such that the
number of maintenance jobs scheduled in any time period must not exceed $K$. 

From a practical point
of view, this is a natural variation of the model. In many real world network maintenance scheduling
problems, there are resource and budget constraints that do not allow too many jobs to be performed at the
same time. For example, the number of crews available to work at night may be limited, or the maintenance operation may require the use of specialized machines, of which very few are available. In the coal supply chain situation that motivated this research, some types of  rail maintenance require the use of such machines: the machines were shared across the whole state, with at most two available in the region at any one time. 
Of course, in practice there can be complicated rules about the combinations of jobs that
are allowed. Disregarding these complications, we propose to study a very simple version of the
model as an abstract combinatorial optimization problem. We also make the simplifying assumptions
that flow is instantaneous, i.e., there are no transit times associated with
the arcs, and that there is always enough flow available to exhaust the network capacity. These are both valid assumptions in the case of the coal supply chain application that motivated this work~\cite{boland2011cap_align-hvcc,boland2012mixed,boland2013capalign-hvcc}. For example, it can be shown that all transit times can be set to zero if all job start times are expressed in a standardized time, in which each job's start time is delayed by the travel time from its location to the port terminal.

The optimization problem is to choose the outage time periods in such a way that the total flow from $s$
to $t$ is maximized. We call this problem Maximum Flow Arc Shutdown Scheduling (\textsc{MFASS}), and
more formally, it can be written as a mixed binary program as follows:
\begin{alignat}{2}
\nonumber \text{maximize}\ z &= \sum_{i=1}^T\sum_{a\in\delta^+(s)}x_{ai}\qquad &&\text{subject to}\\ 
 x_{ai} &\leqslant u_ay_{ai} && a\in J,\ i\in[T], \label{eq:cap_job_arcs} \\
x_{ai} &\leqslant u_a && a\in A\setminus J,\ i\in[T], \label{eq:cap_nonjob_arcs} \\
\sum_{i=1}^Ty_{ai} &= T-1 && a\in J, \label{eq:scheduling} \\
\sum_{a\in\delta^-(v)}x_{ai} &= \sum_{a\in\delta^+(v)}x_{ai} &&  v \in V\setminus \{s,t\},\
i\in[T], \label{eq:flow_conservation}\\
\sum_{a\in J} y_{ai} &\geqslant \lvert J\rvert - K && \ i\in[T],  \label{eq:job_bound}\\
x_{ai} &\geqslant 0 && a\in A,\ i\in[T],  \label{eq:nonnegative_flows}\\
y_{ai} &\in\{0,1\} && a\in A,\ i\in[T], \label{eq:binary_indicators}
\end{alignat}
where $x_{ai}$ for $a \in A$ and $i \in [T]$ denotes the flow on arc $a$ in time period $i$, and
$y_{ai} \in \{0,1\}$ for $a \in A$ and $i \in [T]$ indicates when the arc $a$ is available in time
period $i$, i.e., $y_{ai}=0$ in the period $i$ in which the outage for arc $a$ is scheduled. The
problem is to schedule the maintenance jobs so that the total flow of the network over the time
horizon $T$ is maximized.

In the present work, our focus is not primarily on the real-world application in the background, but
on the abstract optimization problem MFASS and on the properties that make a class of instances
hard or easy. These instance classes may or may not correspond to properties that occur in the coal
supply chain application. For instance, the reduction from \textsc{3-Partition}
in~\cite{boland2012scheduling} shows that the general problem is strongly NP-complete for the class
of instances with $K=3$, and this raises the question about the hardness of the case $K=2$. Nevertheless, the original supply chain application did motivate some features studied. For example, the real-life network is series-parallel, (\cite{boland2012mixed}), some types of maintenance, (especially on the rail network), require the use of scarce equipment, motivating the study of small values of $K$, and the sum of arc capacities entering any node is equal, or nearly equal, to the sum of arc capacities leaving any node, for almost all network nodes (\cite{boland2013unit_time}). In~\cite{boland2013unit_time}
several instance classes for the problem without the job limit per time period were analyzed. 

In
order to classify instances we introduce the following notation. Let $\mathcal C$ be the class of
all \textsc{MFASS} instances. With an upper index $K$ we denote the class of all instances with an
upper bound of $K$ on the number of jobs scheduled per time period, and a lower index indicates
additional restrictions as introduced in~\cite{boland2013unit_time}.
\begin{itemize}
\item Let $\mathcal C_{\text{sp}}$ be the class of instances where the underlying network is series-parallel.
\item Let $\mathcal C_{\text{bal}}$ be the class of instances where the underlying network is
  \emph{balanced}, i.e., for each transshipment node $v\in V\setminus\{s,t\}$ the capacity into this
  node equals the capacity out of this node.
\item Let $\mathcal C_{\text{uc}}$ be the class of \emph{u}nit \emph{c}apacity instances, i.e., the capacities are $u_a=1$ for all arcs $a\in A$.
\item Let $\mathcal C_{\text{aa}}$ be the class of instances where \emph{a}ll \emph{a}rcs have a job associated, i.e., $J=A$.
\end{itemize}
For instance $\left(\mathcal C^3_{\text{sp}}\cap \mathcal C^3_{\text{aa}}\right)\setminus \mathcal
C^3_{\text{bal}}$ is the set of all instances with a series-parallel network which is not balanced,
a job associated with every arc, and the constraint that at most $3$ jobs can be scheduled per time
period. In general, $K$ is not constant, and we also consider instance classes with varying $K$, but
imposing some restrictions on how $K$ can vary relative to other instance parameters. For instance,
$\mathcal C^{\lvert J\rvert}_{\text{sp}}$ is the class of instances with a series-parallel network
and no limit on the number of jobs per time period, and $\mathcal C^{\lvert J\rvert/3}$ contains the
instances in which at most one third of all jobs can be scheduled per time period. As proved
in~\cite{boland2013unit_time}, the classes $\mathcal C^{\lvert J\rvert}_{\text{aa}}$ and $\mathcal
C^{\lvert J\rvert}_{\text{sp}}\cap\mathcal C^{\lvert J\rvert}_{\text{bal}}$ are trivial: it is
always optimal to schedule all jobs at the same time. In contrast, the restriction of the problem to
$\mathcal C^{\lvert J\rvert}_{\text{bal}}$ is still strongly NP-hard, and the restriction to
$\mathcal C^{\lvert J\rvert}_{\text{sp}}$ is NP-hard, but for fixed $T$ it can be solved in
pseudopolynomial time using dynamic programming. Our new complexity results are summarized in Table~\ref{tab:results}.
\begin{table}[htb]
\renewcommand{\arraystretch}{1.2}
  \centering
  \begin{tabular}{@{}ll@{}} \toprule
    Instance class & Complexity \\ \midrule
    $\mathcal C^3_{\text{sp}}\cap \mathcal C^3_{\text{bal}}\cap \mathcal C^3_{\text{aa}}$ & strongly NP-complete (Proposition~\ref{prop:strong_hardness})\\
$\mathcal C^{\lvert J\rvert-1}_{\text{sp}}\cap\mathcal C^{\lvert J\rvert-1}_{\text{bal}}\cap\mathcal C^{\lvert J\rvert-1}_{\text{aa}}$ & NP-complete (Proposition~\ref{prop:weak_hardness}) \\
$\mathcal C_{\text{uc}}$ & NP-complete (Proposition~\ref{prop:unit_cap})\\
$\mathcal C^2$ & $O(\lvert J\rvert^3)$ (Proposition~\ref{prop:poly_k_2}) \\ \bottomrule
  \end{tabular}
  \caption{Complexity results.}
  \label{tab:results}
\end{table}

Note that the classes $\mathcal C_{\text{sp}}$, $\mathcal C_{\text{bal}}$ and $\mathcal
C_{\text{uc}}$ are interesting from the coal chain point of view: the actual network underlying the
work in~\cite{boland2012mixed,boland2012scheduling} is series-parallel, almost balanced, and has the
property that a large proportion of the arcs has the same capacity.

Note that the problem is solvable in polynomial time if both $T$ and $K$ are bounded, say $T\leqslant T_0$ and $K\leqslant K_0$ for some absolute constants $T_0$ and $K_0$. Then $\lvert J\rvert\leqslant K_0T_0$ for any feasible instance, and we can enumerate all partitions of $J$ into at most $T$ sets of size at most $K$ of which there are at most
\[C=\prod_{i=0}^{T_0}\binom{K_0(T_0-i)}{K_0}.\]
For each of these partitions we have to solve $T$ maximum flow problems, hence the run-time is bounded by $CT_0nm$, since the maximum flow problem can be solved in $O(mn)$ time~\cite{king1994faster,orlin2013max}. Consequently, for the asymptotic analysis we are interested in instance classes where at least one of the parameters $T$ and $K$ is unbounded.

The paper is organized as follows. In Section~\ref{sec:K_2} we show that the case $K=2$ can be solved in polynomial time. In addition we provide an explicit description of an optimal solution for $K=2$ and a network with a single transshipment node which leads to a significantly better run-time bound for this case. The hardness results are proved in Section~\ref{sec:hardness}. In Section~\ref{sec:fptas} we present a fully polynomial time approximation scheme for series-parallel networks with fixed time horizon. We also provide a polynomial time approximation scheme for series parallel networks in general when $K = \lvert J\rvert$.

\section{The case $K=2$}\label{sec:K_2}
In this section we consider the case $K=2$. In Section~\ref{subsec:general} we show that this case can be reduced to a maximum weighted matching problem and thus is solvable in polynomial time, and in Section~\ref{subsec:single_node} we give an explicit description of an optimal solution for the case that the network has only a single transshipment node.

\subsection{General networks}\label{subsec:general}
We reduce the problem to a maximum weight perfect matching problem. Let $F_0$ denote the maximum flow value in the whole network, for $a\in J$ let $F_a$ denote the maximum flow when arc $a$ is shut, and for distinct $a,b\in J$ let $F_{ab}$ be the maximum flow when arcs $a$ and $b$ are shut. We set $p=\max\{0,\lvert J\rvert-T\}$ and define an auxiliary graph whose vertex set contains two vertices for every arc $a\in J$ and two sets $W$ and $W'$ of dummy vertices with $\lvert W\rvert=2p$ and $\lvert W'\rvert=2(\lfloor\lvert J\rvert/2\rfloor-p)$. The two vertices for $a\in J$ are denoted by $a$ and $a'$, and the weighted edge set of the auxiliary graph is defined as follows:
\begin{itemize}
\item For distinct arcs $a,b\in J$ there is an edge $\{a,b\}$ with weight $F_{ab}+F_0$.
\item For $a\in J$ there is an edge $\{a,a'\}$ of weight $F_a$.
\item There are all edges of the form $\{a',w\}$ for $a\in J$ and $w\in W\cup W'$. All these edges have zero weight.
\item The vertex set $W'$ induces a matching consisting of zero weight edges.
\end{itemize}
There is a correspondence between perfect matchings in the auxiliary graph and outage schedules. Let $M$ be a perfect matching in the auxiliary digraph. The corresponding schedule has
\begin{itemize}
\item for every edge $\{a,b\}\in M$ with $a,b\in J$ one time period with arcs $a$ and $b$ shut,
\item for every edge $\{a,a'\}\in M$ with $a\in J$ one time period with only arc $a$ shut,
\item all other time periods without shut arcs.
\end{itemize}
This construction is illustrated in Figure~\ref{fig:matching} for $J=\{a,b,\ldots,h\}$ and $T=6$. 
\begin{figure}[htb]
  \centering
  \includegraphics[width=.55\textwidth]{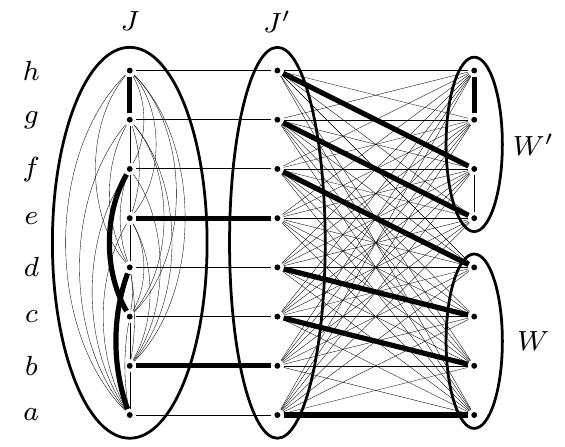}
  \caption{A perfect matching in the auxiliary graph.}
  \label{fig:matching}
\end{figure}
The bold edges form a perfect matching corresponding to scheduling the following outage of schedule: period 1: $\{a,d\}$, period 2: $\{c,f\}$,  period 3: $\{g,h\}$,  period 4: $\{b\}$,  period 5: $\{e\}$,  period 6: $\varnothing$.    

For a perfect matching $M$ we define subsets $M_1\subseteq M$ and $M_2\subseteq M$ by
\begin{align*}
M_1 &= \{\{a,b\}\in M\ :\ a,b\in J\},&
M_2 &= \{\{a,a'\}\in M\ :\ a\in J\}.
\end{align*}
Note that the $2p$ nodes in $W$ must be matched to nodes $a'$, hence
\[\lvert M_2\rvert\leqslant \lvert J\rvert-2p\leqslant\lvert J\rvert-2(\lvert J\rvert-T)=2T-\lvert J\rvert,\]
and with $\lvert M_1\rvert=\frac12(\lvert J\rvert-\lvert M_2\rvert)$ this implies
\[\lvert M_1\rvert+\lvert M_2\rvert=\frac12(\lvert J\rvert-\lvert M_2\rvert)+\lvert M_2\rvert=\frac12(\lvert J\rvert+\lvert M_2\rvert)\leqslant T.\]
The total throughput for the schedule corresponding to the matching $M$ is
\begin{multline*}
\sum_{\{a,b\}\in M_1}F_{ab}+\sum_{\{a,a'\}\in M_2}F_a+(T-\lvert M_1\rvert-\lvert M_2\rvert)F_0\\
=\sum_{\{a,b\}\in M_1}\left(F_{ab}+F_0\right)+\sum_{\{a,a'\}\in M_2}F_a+(T-2\lvert M_1\rvert-\lvert M_2\rvert)F_0 = \omega(M)+(T-\lvert J\rvert)F_0,
\end{multline*}
where $\omega(M)$ is the weight of $M$. Thus the original problem is equivalent to finding a maximum weighted perfect matching in the auxiliary graph, and with an efficient implementation~\cite{gabow1990data} of the blossom algorithm~\cite{edmonds1965paths} we have proved the following proposition.
\begin{proposition}\label{prop:poly_k_2}
For $K=2$ the problem \textsc{MFASS} can be solved in $O(\lvert J\rvert^3)$ time.
\end{proposition}

\subsection{The single node case}\label{subsec:single_node}
Consider a network with a single transshipment node $v$, a job set $J$, a time horizon $T$ and $K=2$. We use the notation $J^-=\delta^-(v)\cap J$ and $J^+=\delta^+(v)\cap J$ and assume without loss of generality that $\lvert J^-\rvert\leqslant\lvert J^+\rvert$. We order the arcs in $J^-$ and $J^+$ such that the capacities are non-increasing, i.e. $J^-=\{a_1,\ldots,a_r\}$ and $J^+=\{b_1,\ldots,b_s\}$ ($s\geqslant r$) with
\begin{align*}
  u_{a_1}&\geqslant u_{a_2}\geqslant\cdots\geqslant u_{a_r}, & u_{b_1}&\geqslant u_{b_2}\geqslant\cdots\geqslant u_{b_s}.
\end{align*}
Note that it is necessary for feasibility that $r+s\leqslant 2T$, and in particular $r\leqslant T$. We will show that an optimal solution can be obtained as follows.
\begin{proposition}\label{prop:K_2_single_node}
An optimal solution for the single node problem with $K=2$ is given by the following schedule.
\begin{itemize}
\item For $i=1,2,\ldots,r$ take arcs $a_i$ and $b_i$ out in time period $i$.
\item For $i=r+1,r+2,\ldots,\min\{T,\,2T-s\}$ take arc $b_i$ out in time period $i$.
\item If $s>T$ then for $i=2T-s+1,2T-s+2,\ldots,T$ take arcs $b_i$ and $b_{2T+1-i}$ out in time
  period $i$.
\end{itemize}
\end{proposition}
For the proof of Proposition~\ref{prop:K_2_single_node} we will need the following notation for the
inbound and outbound capacities under various outage scenarios.
\begin{align*}
X&= \sum_{a\in\delta^-(v)}u_a, & Y&= \sum_{a\in\delta^+(v)}u_a, \\
X_i &= X-u_{a_i}\text{ for }1\leqslant i\leqslant r, & Y_i &= Y-u_{b_i}\text{ for }1\leqslant
i\leqslant s, \\
X_{ij} &= X-u_{a_i}-u_{a_j}\text{ for }1\leqslant i<j\leqslant r, & Y_{ij} &=
Y-u_{b_i}-u_{b_j}\text{ for }1\leqslant i<j\leqslant s.
\end{align*}
We need the following inequality.
\begin{lemma}\label{lem:inequality}
For any real numbers $x_1,\ldots,x_6$ satisfying $x_3,x_4\in[x_1,\,x_2]$, $x_3+x_4=x_1+x_2$ and $x_5\leqslant x_6$, we
 have
\[\min\{x_3,\,x_6\}+\min\{x_4,\,x_5\}\geqslant\min\{x_1,\,x_6\}+\min\{x_2,\,x_5\}.\]
\end{lemma}
\begin{proof}
The LHS is $\min\{x_3+x_4,\,x_3+x_5,\,x_6+x_4,\,x_6+x_5\}$, and we have
\begin{align*}
  x_3+x_4 = x_1+x_2 &\geqslant \min\{x_1,\,x_6\}+\min\{x_2,\,x_5\},\\
  x_3+x_5 \geqslant x_1+x_5 &\geqslant \min\{x_1,\,x_6\}+\min\{x_2,\,x_5\},\\ 
  x_4+x_6 \geqslant x_1+x_5 &\geqslant \min\{x_1,\,x_6\}+\min\{x_2,\,x_5\}, \ \ \text{and}\\ 
  x_6+x_5 &\geqslant \min\{x_1,\,x_6\}+\min\{x_2,\,x_5\}. \qedhere
\end{align*}
\end{proof}
\begin{proof}[Proof of Proposition~\ref{prop:K_2_single_node}.] 
Let $S$ be the schedule described in the proposition, and let $S_i$ be the set of arcs that are
scheduled to be shut in period $i$ ($i=1,\ldots,T$). For the sake of contradiction, suppose that $S$ is not optimal. Among all optimal
schedules we can choose one, say $S'$, that differs from $S$ as late as possible, i.e., such that
the smallest index $i$ with $S'_i\neq S_i$ is maximal, where $S'_i$ is the set of arcs that are shut down in
period $i$ according to schedule $S'$.
\begin{description}
\item[Case 1.] $i\leqslant r$. There are indices $p,q\geqslant i$ with $a_i\in S'_p$, and $b_i\in
  S'_q$. Without loss of generality, we may assume $p=i$, since otherwise $S'_i$ could be swapped
  with $S'_p$ to yield a schedule with the same objective value. Furthermore, $q>i$ since otherwise $S'_i=S_i$. Replacing $S'_i$ with $\{a_i,\,b_i\}$ and $S'_q$ with $S'_i\cup
  S'_q\setminus\{a_i,\,b_i\}$ we obtain another schedule $S''$ which agrees with $S$ for one time
  period more than $S'$. In order to arrive at the required contradiction we have to check that schedule $S''$ is not worse than schedule $S'$. Note
  that the schedules $S'$ and $S''$ differ only in periods $i$ and $q$. We distinguish several cases
  for the sets $S'_i$ and $S'_q$. For each case we write down the total flows in periods $i$ and
  $q$ for the schedules $S'$ and $S''$, and then we apply Lemma~\ref{lem:inequality} to verify that
  $S''$ is at least as good as $S'$. 
\begin{description}
\item[Case 1.1.] $S'_i=\{a_i,\,b_k\}$ and $S'_q=\{a_j,\,b_i\}$ for some $j\in\{i+1,\ldots,r\}$ and
    $k\in\{i+1,\ldots,s\}$.
    \begin{align*}
      S' :\ \min\{X_i,\,Y_k\} &+ \min\{X_j,\,Y_i\}, &
      S'':\ \min\{X_i,\,Y_i\} &+ \min\{X_j,\,Y_k\}.
    \end{align*}
The claim follows from Lemma~\ref{lem:inequality} with 
$(x_1,\ldots,x_6)=(X_i,\,X_j,\,X_j,\,X_i,\,Y_i,\,Y_k)$.
\item[Case 1.2.] $S'_i=\{a_i,\,b_k\}$ for some $k\in\{i+1,\ldots,s\}$, and $S'_q=\{b_i\}$. 
    \begin{align*}
      S' :\ \min\{X_i,\,Y_k\} &+ \min\{X,\,Y_i\}, &
      S'':\ \min\{X_i,\,Y_i\} &+ \min\{X,\,Y_k\}.
    \end{align*}
The claim follows from Lemma~\ref{lem:inequality} with 
$(x_1,\ldots,x_6)=(X_i,\,X,\,X,\,X_i,\,Y_i,\,Y_k)$.
\item[Case 1.3.] $S'_i=\{a_i\}$ and $S'_q=\{a_j,\,b_i\}$ for some $j\in\{i+1,\ldots,r\}$.
    \begin{align*}
      S' :\ \min\{X_i,\,Y\} &+ \min\{X_{j},\,Y_i\}, &
      S'':\ \min\{X_i,\,Y_i\} &+ \min\{X_{j},\,Y\}.
    \end{align*}
The claim follows from Lemma~\ref{lem:inequality} with 
$(x_1,\ldots,x_6)=(X_i,\,X_j,\,X_j,\,X_i,\,Y_i,\,Y)$.
\item[Case 1.4.] $S'_i=\{a_i\}$ and $S'_q=\{b_i\}$. 
    \begin{align*}
      S' :\ \min\{X_i,\,Y\} &+ \min\{X,\,Y_i\}, &
      S'':\ \min\{X_i,\,Y_i\} &+ \min\{X,\,Y\}.
    \end{align*}
The claim follows from Lemma~\ref{lem:inequality} with
$(x_1,\ldots,x_6)=(X_i,\,X,\,X,\,X_i,\,Y_i,\,Y)$.
\item[Case 1.5.] $S'_i=\{a_i,\,a_j\}$ and $S'_q=\{b_i,\,b_k\}$. 
    \begin{align*}
      S' :\ \min\{X_{ij},\,Y\} &+ \min\{X,\,Y_{ik}\}, &
      S'':\ \min\{X_i,\,Y_i\} &+ \min\{X_{j},\,Y_{k}\}.
    \end{align*}
The claim follows from Lemma~\ref{lem:inequality} applied twice, first with
$(x_1,\ldots,x_6)=(X_{ij},\,X,\,X_j,\,X_i,\,Y_i,\,Y_k)$ and then with
$(x_1,\ldots,x_6)=(Y_{ik},\,Y,\,Y_i,\,Y_k,\,X_{ij},\,X)$:
\[\min\{X_j,\,Y_k\}+\min\{X_i,\,Y_i\} \geqslant \min\{X_{ij},\,Y_k\}+\min\{X,\,Y_{i}\} \geqslant \min\{X_{ij},\,Y\}+\min\{X,\,Y_{ik}\}.\]
\item[Case 1.6.] $S'_i=\{a_i,\,a_j\}$ for some $j\in\{i+1,\ldots,r\}$, and $S'_q=\{b_i\}$. 
    \begin{align*}
      S' :\ \min\{X_{ij},\,Y\} &+ \min\{X,\,Y_{i}\}, &
      S'':\ \min\{X_i,\,Y_i\} &+ \min\{X_{j},\,Y\}.
    \end{align*}
The claim follows from Lemma~\ref{lem:inequality} with
$(x_1,\ldots,x_6)=(X_{ij},\,X,\,X_j,\,X_i,\,Y_i,\,Y)$.
\item[Case 1.7.] $S'_i=\{a_i\}$ and, $S'_q=\{b_i,\,b_k\}$ for some $k\in\{i+1,\ldots,s\}$. 
    \begin{align*}
      S' :\ \min\{X_{i},\,Y\} &+ \min\{X,\,Y_{ik}\}, &
      S'':\ \min\{X_i,\,Y_i\} &+ \min\{X,\,Y_{k}\}.
    \end{align*}
The claim follows from Lemma~\ref{lem:inequality} with
$(x_1,\ldots,x_6)=(Y_{ik},\,Y,\,Y_k,\,Y_i,\,X_i,\,X)$.
\item[Case 1.8.] $S'_i=\{a_i,\,a_j\}$ and $S'_q=\{a_k,b_i\}$ for some $j,k\in\{i+1,\ldots,r\}$.
    \begin{align*}
      S' :\ \min\{X_{ij},\,Y\} &+ \min\{X_k,\,Y_{i}\}, &
      S'':\ \min\{X_i,\,Y_i\} &+ \min\{X_{jk},\,Y\}.
    \end{align*}
The claim follows from Lemma~\ref{lem:inequality} with
$(x_1,\ldots,x_6)=(X_{ij},\,X_k,\,X_{jk},\,X_i,\,Y_i,\,Y)$.
\item[Case 1.9.] $S'_i=\{a_i,\,b_j\}$ and $S'_q=\{b_k,\,b_i\}$ for some $j,k\in\{i+1,\ldots,s\}$.
    \begin{align*}
      S' :\ \min\{X_{i},\,Y_j\} &+ \min\{X,\,Y_{ik}\}, &
      S'':\ \min\{X_i,\,Y_i\} &+ \min\{X,\,Y_{jk}\}.
    \end{align*}
The claim follows from Lemma~\ref{lem:inequality} with
$(x_1,\ldots,x_6)=(Y_{ik},\,Y_j,\,Y_{jk},\,Y_i,\,X_i,\,X)$.
\end{description}
\item[Case 2.] $i>r$ and $S_i=\{b_i\}$. Without loss of generality, we assume that $b_i\in S'_i$, and then
  $S'_i\neq S_i$ implies $S'_i=\{b_i,\,b_j\}$ for some $j\in\{i+1,\ldots,s\}$. Furthermore,
  $S'_{i}\cup S'_{i+1}\cup\cdots\cup S'_{T}=S_{i}\cup S_{i+1}\cup\cdots\cup S_{T}$, and from $\lvert S_i\rvert=1$ and $\lvert S'_i\rvert=2$ it follows that $\lvert S'_q\rvert\leqslant 1$ for some
  $q\in\{i+1,\ldots,T\}$. Consequently, $S'_q=\emptyset$ or $S'_q=\{b_k\}$ for some $k\in\{i+1,\ldots,s\}$. Replacing $S'_i$ with
  $\{b_i\}$ and $S'_q$ with $\{b_j\}\cup S'_q$ we obtain another schedule $S''$ which agrees with
  $S$ for one time period more than $S'$, and we claim that $S''$ is not worse than $S'$. If $S'_q=\{b_k\}$ then the total flows in periods $i$ and $q$ are
\begin{align*}
      S' :\ \min\{X,\,Y_{ij}\} &+ \min\{X,\,Y_{k}\}, &
      S'':\ \min\{X,\,Y_i\} &+ \min\{X,\,Y_{jk}\},
    \end{align*}
and the claim follows from Lemma~\ref{lem:inequality} with
$(x_1,\ldots,x_6)=(Y_{ij},\,Y_k,\,Y_{jk},\,Y_i,\,X,\,X)$.
If $S'_q=\emptyset$ then the total flows in periods $i$ and $q$ are
\begin{align*}
      S' :\ \min\{X,\,Y_{ij}\} &+ \min\{X,\,Y\}, &
      S'':\ \min\{X,\,Y_i\} &+ \min\{X,\,Y_{j}\},
    \end{align*}
and the claim follows from Lemma~\ref{lem:inequality} with
$(x_1,\ldots,x_6)=(Y_{ij},\,Y,\,Y_j,\,Y_i,\,X,\,X)$.
\item[Case 3.] $i>r$ and $S_i=\{b_i,\,b_{2T+1-i}\}$. We have $S'_i\cup\cdots\cup
  S'_T=S_i\cup\cdots\cup S_T=\{b_i,\,b_{i+1},\ldots,b_{2T+1-i}\}$. This implies $\lvert S_p\rvert=2$
  for all $p\in\{i,\ldots,T\}$. Without loss of generality, we assume 
  $S'_i=\{b_i,\,b_j\}$ for some $j\in\{i+1,\ldots,2T-i\}$, and there exists $q\in\{i+1,\ldots,T\}$ with $S'_q=\{b_k,b_\ell\}$ for
  $\ell=2T+1-i$ and some $k\in\{i+1,\ldots,2T-i\}$. Replacing $S'_i$ with
  $\{b_i,\,b_\ell\}$ and $S'_q$ with $\{b_j,\,b_k\}$ we obtain another schedule $S''$ which agrees with $S$
  for one time period more than $S'$. The total flows in periods $i$ and $q$ are
\begin{align*}
      S' :\ \min\{X,\,Y_{ij}\} &+ \min\{X,\,Y_{k\ell}\}, &
      S'':\ \min\{X,\,Y_{i\ell}\} &+ \min\{X,\,Y_{jk}\}.
    \end{align*}
From Lemma~\ref{lem:inequality} with
$(x_1,\ldots,x_6)=(Y_{ij},\,Y_{k\ell},\,Y_{i\ell},\,Y_{jk},\,X,\,X)$
it follows that $S''$ is at least as good as $S'$ and this is the required contradiction.\qedhere 
\end{description}
\end{proof}
Since sorting the arcs dominates the run-time of the algorithm to find the solution described in Proposition~\ref{prop:K_2_single_node} we obtain the following stronger run-time bound for the single-node case.
\begin{corollary}\label{cor:k=2_single_node}
For $K=2$ and a single transshipment node \textsc{MFASS} can be solved in time $O(\lvert J\rvert\log\lvert J\rvert)$.
\end{corollary}

\section{Hardness results}\label{sec:hardness}
Before proving the hardness results we make precise the definition of \emph{series-parallel network}. In the present paper this term refers to a \emph{two-terminal series-parallel network}: a network that has a single source and single sink and is constructed by a sequence of series and parallel compositions starting from single arcs. For two networks $N_1$ and $N_2$ the \emph{parallel composition} of $N_1$ and $N_2$ is obtained by identifying the source node $s_1$ and sink node $t_1$ of $N_1$ with the source node $s_2$ and sink node $t_2$ of $N_2$, respectively. The \emph{series composition} of $N_1$ and $N_2$ is obtained by identifying the sink node $t_1$ of $N_1$ with the source node $s_2$ of $N_2$. The construction of a series parallel network can be encoded into a tree, the so-called SP-tree, whose leaves are the arcs of the network. This is illustrated in Figure~\ref{fig:sp_tree}.
\begin{figure}[htb]
\centering
  \begin{minipage}{.3\linewidth}
\includegraphics[width=\textwidth]{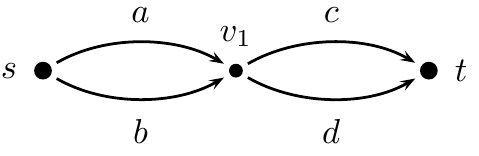}
  \end{minipage}\qquad
  \begin{minipage}{.3\linewidth}
\includegraphics[width=\textwidth]{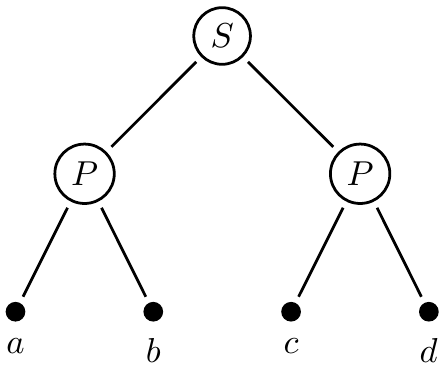}
  \end{minipage}
  \caption{A series-parallel network and the corresponding SP-tree.}
  \label{fig:sp_tree}
\end{figure}
\begin{proposition}\label{prop:strong_hardness}
The restriction of \textsc{MFASS} to the instance class $\mathcal C^3_{\text{sp}}\cap \mathcal C^3_{\text{bal}}\cap \mathcal C^3_{\text{aa}}$ is strongly NP-complete.
\end{proposition}
\begin{proof}
We use reduction from \textsc{3-Partition}. Let a \textsc{3-Partition} instance be given by an
integer $B$ and a set $\{u_1,\ldots,u_{3n}\}$ of integers with $B/4<u_j<B/2$ for all $j$ and
$\sum_{j=1}^{3n}u_j=nB$. The problem is to decide if there is a partition of the set
$\{u_1,\ldots,u_{3n}\}$ into $n$ triples such that the sum of each triple equals $B$. We define new
numbers $u'_i$ for $i=1,\ldots,3n$ by $u'_i=3u_i-B$. Note that 
\begin{equation}\label{eq_zerosum}
\sum_{i=1}^{3n}u'_i=\sum_{i=1}^{3n}\left(3u_i-B\right)=3\sum_{i=1}^{3n}u_i-3nB=3nB-3nB=0,
\end{equation}
and for every triple $(i,j,k)$ we have
\begin{multline*}
u'_i+u'_j+u'_k=0 \iff \left(3u_i-B\right)+\left(3u_j-B\right)+\left(3u_k-B\right)=0 \\
\iff 3(u_i+u_j+u_k)-3B=0\iff u_i+u_j+u_k=B .
\end{multline*}
Without loss of generality we assume that for some integer $r$, we have $u'_i\geqslant 0$ for $i\leqslant r$ and $u'_i<0$ for $i>r$. We define an instance of our problem with $K=3$, $T=n$, a single transshipment node $v$ and the following arcs:
\begin{itemize}
\item For $i=1,2,\ldots,r$ there is an arc $a_i$ into $v$ having capacity $u'_i$, and
\item for $i=r+1,\ldots,3n$ there is an arc $a_i$ that goes out of $v$ and has capacity $-u'_i$.
\end{itemize}
This is illustrated in Figure~\ref{fig:strong_111}, where the arc labels represent capacities and
all arcs have an associated job, i.e., $J=A$. Obviously the network is series-parallel. From $K=3$, $T=n$ and $\lvert J\rvert=3n$ it follows
that we need to shut down exactly 3 arcs in every period. It follows from~\eqref{eq_zerosum} that the network
is balanced. Let $X=u_1+\ldots+u_r$ be the capacity of the network. Clearly, $(n-1)X$ is an upper
bound for the total throughput, and we claim that this bound can be achieved if and only if the set $\{u'_i\ :\
i=1,\ldots,3n\}$ can be partitioned into triples that sum up to zero, or equivalently, the set
$\{u_i\ :\ i=1,\ldots,3n\}$ can be partitioned into triples that sum up to $B$. First assume that 
\[\{1,\ldots,3n\}=\{i_1,\,j_1,\,k_1\}\cup\cdots\cup\{i_n,\,j_n,\,k_n\}\]
is a partition with $u'_{i_p}+u'_{j_p}+u'_{k_p}=0$ for all $p\in\{1,\ldots,n\}$. Consider the
schedule that shuts down the arcs $a_{i_p}$, $a_{j_p}$ and $a_{k_p}$ in period $r$. It follows from
$u'_{i_p}+u'_{j_p}+u'_{k_p}=0$ that the network with arc set
$A_p=A\setminus\left\{a_{i_p},\,a_{j_p},\,a_{k_p}\right\}$ is balanced, and therefore we get a feasible flow in which every arc in
$A_p$ is at capacity. Therefore, every arc is at capacity in $n-1$ periods and the total throughput equals
\[\sum_{i=1}^r(n-1)u_i=(n-1)X.\]
Conversely, if there is a schedule with a total throughput of $(n-1)X$ then every arc must be at
capacity in every period in which it is not shut down. This implies that in every period $p\in\{1,\ldots,n\}$ the network
with arc set $A_p=A\setminus\left\{a_{i_p},\,a_{j_p},\,a_{k_p}\right\}$,  is balanced, where $i_p$,
$j_p$ and $k_p$ are the indices of the arcs that are shut down in period $p$. Consequently
$u'_{i_p}+u'_{j_p}+u'_{k_p}=0$ for every $p\in\{1,\ldots,n\}$, and this yields a solution for the
\textsc{3-Partition} instance.  
\end{proof}
\begin{figure}[htb]
\centering
    \begin{minipage}[b]{.49\linewidth}
      \includegraphics[width=\textwidth]{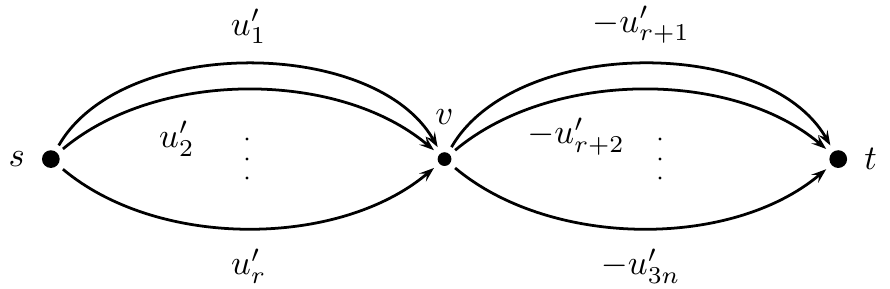}
      \caption{The network for $\mathcal C^3_{\text{sp}}\cap \mathcal C^3_{\text{bal}}\cap \mathcal C^3_{\text{aa}}$.}\label{fig:strong_111}
    \end{minipage}\hfill
    \begin{minipage}[b]{.49\linewidth}
      \includegraphics[width=\textwidth]{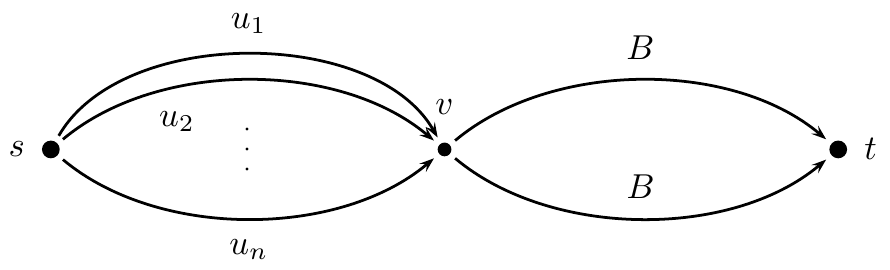}
      \caption{The network for $\mathcal C^{\lvert J\rvert-1}_{\text{sp}}\cap\mathcal C^{\lvert J\rvert-1}_{\text{bal}}\cap\mathcal C^{\lvert J\rvert-1}_{\text{aa}}$.}\label{fig:weak_111}
    \end{minipage}
\end{figure}
\begin{proposition}\label{prop:weak_hardness}
The restriction of \textsc{MFASS} to the instance class $\mathcal C^{\lvert J\rvert-1}_{\text{sp}}\cap\mathcal C^{\lvert J\rvert-1}_{\text{bal}}\cap\mathcal C^{\lvert J\rvert-1}_{\text{aa}}$ is NP-complete.
\end{proposition}
\begin{proof}
We use reduction from \textsc{Partition}. Let a \textsc{Partition} instance be given by an integer
$B$ and a set $\{u_1,\ldots,u_{n}\}$ of integers with $\sum_{j=1}^{n}u_j=2B$. The problem is to
decide if there is a partition of the set $\{u_1,\ldots,u_{n}\}$ into two parts such that the sum of
each part equals $B$. The network used for the reduction is shown in Figure~\ref{fig:weak_111},
where the arc labels represent capacities and all arcs have an associated job, i.e., $J=A$. Consider
this network for the time horizon $T=2$ and with $K=n+1=\lvert J\rvert-1$. Each of the two arcs of capacity $B$ can
carry at most $B$ units of flow over the whole time horizon, because it needs to be shut down for
one period. Therefor $2B$ is an upper bound for the total throughput. It is not possible to have a
flow of $2B$ in a single period, since otherwise all $n+2$ arcs would need to be shut in the other
period. Therefore, in order to achieve the bound of $2B$ we must have a flow of value $B$ in each time period. This is possible if and only if the total capacity of the arcs between $s$ and $v$ that are shut
down in period 1 is $B$, i.e., the \textsc{Partition} instance is a YES instance.
\end{proof}
Note that the algorithm from~\cite{boland2013unit_time} for series-parallel networks and $K=\lvert J\rvert$ which is pseudopolynomial for fixed $T$ can be adapted to the case $K=\lvert J\rvert-1$. This algorithm computes a list of $T$-dimensional vectors for each node of the SP-tree. The vectors at a node $v$ of the SP-tree represent the possible throughputs for the corresponding subnetwork: $(z_1,\ldots,z_T)$ is in the list at node $v$ if and only if the jobs for arcs in the subnetwork can be scheduled such that the maximum flow value for the subnetwork in time period $i$ is $z_i$ ($i=1,\ldots,T$). In each node of the tree we flag a vector that can only be achieved by scheduling all jobs at the same time (which is at most one per node in the tree). Finally, when we scan the list at the root node in order to determine the optimal solution, we exclude the flagged vector.

In~\cite{boland2013unit_time}, the class $\mathcal C_{\text{uc}}$ of instances where every arc has unit capacity was shown to be tractable when there is no limit for the number of jobs per time period. We finish this section with a proof that this class becomes NP-complete when such a limit is introduced.
\begin{proposition}\label{prop:unit_cap}
  The restriction of \textsc{MFASS} to the instance class $\mathcal C_{\text{uc}}$ is NP-complete.
\end{proposition}
\begin{proof}
We use reduction from \textsc{3-Partition}. Let a \textsc{3-Partition} instance be given by an integer $B$ and a set $\{u_1,\ldots,u_{3n}\}$ of integers with $B/4<u_j<B/2$ for all $j$ and $\sum_{j=1}^{3n}u_j=nB$. This can be reduced to the instance presented in Figure~\ref{fig:unit_cap}, where every arc has unit capacity and the set $J$ is represented by dashed arcs.
\begin{figure}[htb]
  \centering
  \includegraphics{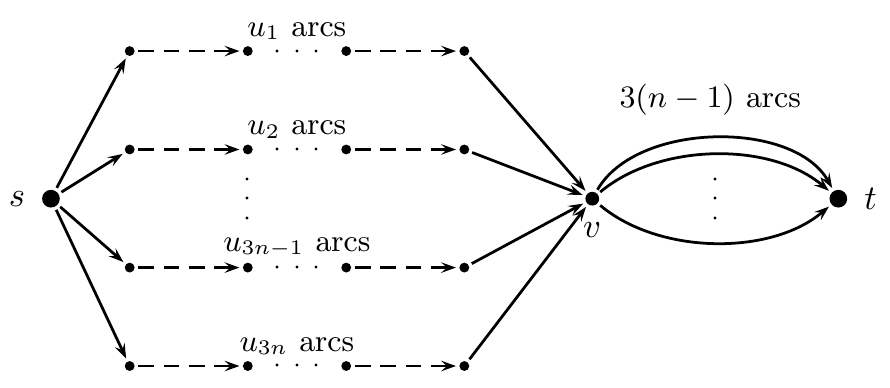}
  \caption{Instance for the reduction in the proof of Proposition~\ref{prop:unit_cap}. The dashed arcs indicate the set $J$ of arcs with an associated job.}
  \label{fig:unit_cap}
\end{figure}
Since \textsc{3-Partition} is strongly NP-hard we may assume that the numbers $u_i$ are bounded by a
polynomial in the input size, and this ensures that the network size is polynomial in the size of
the \textsc{3-Partition} instance. We consider this network with a time horizon $T=n$ and a bound of
$K=B$ jobs per time period. The total throughput is bounded by $3n(n-1)$ since the total capacity of
the arcs entering node $t$ is $3(n-1)$ and there are $n$ time periods. From $\lvert J\rvert=nB$ it follows that exactly $B$ jobs have to be
scheduled in each time period.  We claim that the bound of $3n(n-1)$ on the total throughput can
be achieved if and only if the \textsc{3-Partition} instance is a YES instance.
First suppose the \textsc{3-Partition} instance is a YES instance, and let
\[\{1,\ldots,3n\}=\{i_1,j_1,k_1\}\cup\cdots\cup \{i_n,j_n,k_n\}\] 
be a partition with $u_{i_p}+u_{j_p}+u_{k_p}=B$ for all $p\in\{1,\ldots,n\}$. We obtain a schedule
that achieves the upper bound as follows. In time period $p$ we shut down the arcs on the paths
number $i_p$, $j_p$ and $k_p$, where the paths between $s$ and $v$ are numbered
from top to bottom in Figure~\ref{fig:unit_cap}, i.e., the $i$-th path contains exactly $u_i$ dashed
arcs. Conversely, suppose that there is a schedule that
achieves a total throughput of $3n(n-1)$. For $p\in\{1,\ldots,n\}$ let $I_p$ be the set of paths on which at least one arc is shut down in period
$p$. In order to achieve a total throughput of $3n(n-1)$ we must have a flow of value $3(n-1)$ in
each time period. Therefore, in each period we can shut down arcs on at most 3 paths from $s$ to
$v$, i.e., $\lvert I_p\rvert\leqslant 3$ for all $p\in\{1,\ldots,n\}$. Since all dashed arcs have to
be shut down in some time period we have $I_1\cup\cdots\cup I_n=\{1,\ldots,3n\}$, and consequently,
$\lvert I_p\rvert=3$ for all $p$ and $I_p\cap I_{p'}=\emptyset$ for all $p\neq p'$. This implies
that in every time period all arcs on exactly 3 paths are shut down, hence $\sum_{i\in I_p}u_i=B$
for every $p\in\{1,\ldots,n\}$ and the 3-sets $I_1,\ldots,I_n$ form a solution of the \textsc{3-Partition} instance.
\end{proof}

\section{An FPTAS for series-parallel networks with fixed T}\label{sec:fptas}
In this section we restrict our attention to series-parallel networks. We modify the algorithm from~\cite{boland2013unit_time} such that the bound $K$ can be taken into account. For fixed time horizon $T$, this algorithm runs in pseudopolynomial time, and we use it together with scaling and rounding~\cite{williamson2011design} to design an FPTAS. 

The algorithm presented in~\cite{boland2013unit_time} starts at the leaves of the SP-tree and computes a list of vectors $z=(z_1,\ldots,z_T)$ for each node of the SP-tree, where the list at a node $v$ in the SP-tree contains exactly the vectors $z$ such that there exists some schedule for which the subnetwork corresponding to $v$ can carry flow $z_i$ in time period $i$ for $i=1,\ldots,T$. In the problem variant studied in~\cite{boland2013unit_time} there is no restriction on the number of arcs that can be shut in a period, so it is sufficient to keep track of the possible flow vectors at the nodes of the SP-tree. But the same capacity vector can be realised through different schedules. For instance, for the network shown in in Figure~\ref{fig:example}, there are three possibilities to get the flow vector $(7,0)$, i.e. 7 units in the first time period and zero flow in the second period:
\begin{itemize}
\item shut 2 arcs in period 1 (arcs with capacities 1 and 2), and 2 arcs in period 2 (arcs with capacities 8 and 7); or
\item shut 1 arc in period 1 (arc with capacity 1 or 2), and 3 arcs in period 2 (arcs with capacities 8, 7 and (2 or 1)); or
\item shut no arc in period 1, and all four arcs in period 2.
\end{itemize}
Thus with a limit $K$ for the number of shut arcs per time period it becomes important to keep track
of the number of arcs shut in each period along with maximum flow that can be sent in that
period. Let $j_i$ represent the number of arcs shut in the $i^{th}$ period. We determine lists of
\emph{job-capacity} vectors of the form  $z=((j_1,z_1),(j_2,z_2),\ldots,(j_T,z_T))$ at each node of the SP-tree. The interpretation of such a vector $z$ in the list of node $N$ is that there is a solution in which, for $i=1,\ldots,T$, in time period $i$ exactly $j_i$ arcs from the subnetwork corresponding to $N$ are shut, and this subnetwork has capacity $z_i$. Due to the symmetry with respect to the time periods it is no loss of generality to require the job-capacity vectors to be ordered. Hence we consider only vectors that satisfy, for $i=1,\ldots,T-1$, either $z_i>z_{i+1}$ or $z_i=z_{i+1}$ and $j_i\geqslant j_{i+1}$. We say that a vector with this property is in \emph{standard form}, and we note  that for every job-capacity vector there is a unique vector in standard form which can be obtained by a permutation of the entries. The list at a leaf node of the tree, corresponding to an arc $a$ of the network, consists of the unique vector $((0,u_a),(0,u_a),\ldots,(0,u_a),(1,0))$ if $a \in J$ or $((0,u_a),(0,u_a),\ldots,(0,u_a),(0,u_a))$ if $a \notin J$. As in~\cite{boland2013unit_time}, let $\mathcal L$ and $\mathcal W$ denote the sets of leaves and internal nodes of the  SP-tree, and let $\mathcal W_i$ ($i=0,\ldots,d$) be the set of internal nodes at distance $i$ from the root. The lists of job-capacity vectors are computed as described in Algorithm~\ref{alg:sp_networks_K}.
\begin{algorithm}[htb]
  \caption{Maximizing total throughput for series-parallel networks under uniform maintenance limit $K$}\label{alg:sp_networks_K}
  \begin{tabbing}
    .....\=.....\=.....\=.....\=.....\=................... \kill \\
\textbf{for} $v\in\mathcal L$ \textbf{do} \\
\> Let $a\in A$ be the arc corresponding to $v$\\
\> \textbf{if} $a\in J$ \textbf{then} $L_v\leftarrow[((0,u_a),(0,u_a),\ldots,(0,u_a),(1,0))]$\\ 
\> \textbf{else} $L_v\leftarrow[((0,u_a),(0,u_a),\ldots,(0,u_a),(0,u_a))]$\\
\textbf{for} $i=d,d-1,\ldots,0$ \textbf{do}\\
\> \textbf{for} $v\in\mathcal W_i$ \textbf{do}\\
\> \> $L_v\leftarrow[]$\qquad\{initialize empty list\}\\
\> \> Let $u$ and $w$ be the child nodes of $v$\\
\> \> \textbf{for} $(z,z')\in L_u\times L_w$ and $\pi$ permutation of $\{1,2\ldots,T\}$ \textbf{do}\\
\> \> \> \textbf{for} $i\in[T]$ \textbf{do} $j''_{i}=j_{i}+j'_{\pi(i)}$\\
\> \> \> \textbf{if} $j''_{i}\leqslant K$ for all $i\in[T]$ \textbf{then}\\
\> \> \> \> \textbf{if} $v$ is a parallel composition node \textbf{then}\\
\> \> \> \> \> \textbf{for} $i\in[T]$ \textbf{do} $z''_{i}=z_{i}+z'_{\pi(i)}$\\
\> \> \> \> \textbf{else}\\
\> \> \> \> \> \textbf{for} $i\in[T]$ \textbf{do} $z''_{i}=\min\{z_{i},\,z'_{\pi(i)}\}$\\
\> \> \> sort $z''$ to get the corresponding canonical vector\\
\> \> \> \textbf{if} $z''\not\in L_v$ \textbf{then} add $z''$ to $L_v$\\
Let $v$ be the root node\\
return $\max\limits_{z\in L_v}\sum\limits_{i=1}^Tz_{i}$
  \end{tabbing}
\end{algorithm}
\begin{example} 
Consider the series-parallel graph in Figure~\ref{fig:example} where arc labels indicate capacities, all arcs need maintenance for a period over a time horizon of 2 periods. Suppose that $K=3$. In Figure~\ref{fig:algorithm_run_k}, we show how job-capacity vectors are computed in the SP-tree.
 \end{example}
\begin{figure}[htb]
  \begin{minipage}[b]{.45\linewidth}
\includegraphics[width=\textwidth]{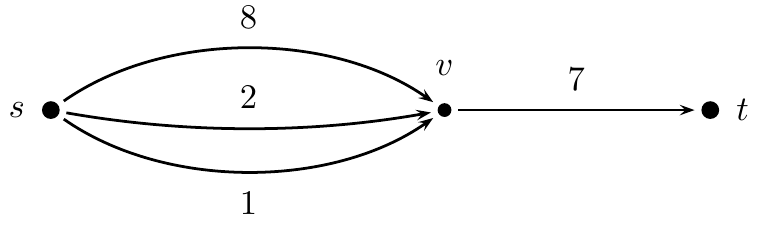}
\vspace{1cm}
\caption{Example network.}\label{fig:example}    
  \end{minipage}\hfill
  \begin{minipage}[b]{.53\linewidth}
\includegraphics[width=\textwidth]{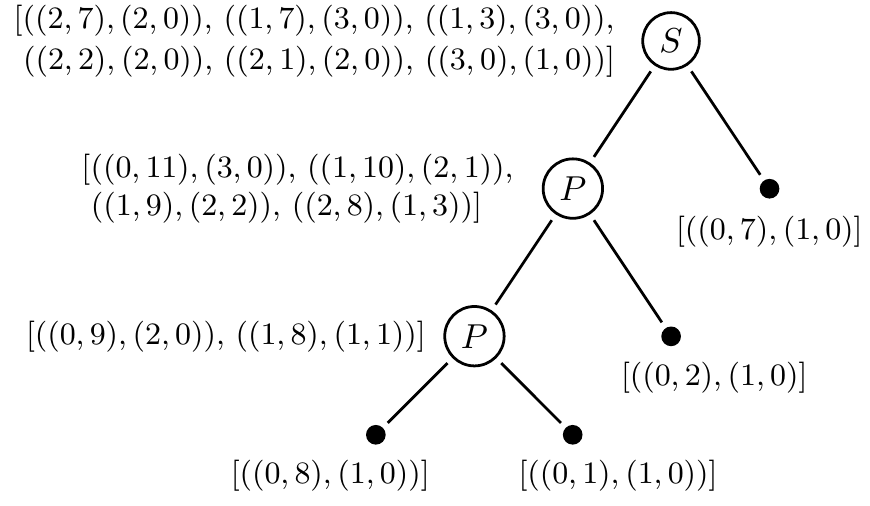}
  \caption{Computation of job-capacity vectors.}
  \label{fig:algorithm_run_k}
\end{minipage}
\end{figure}
\begin{proposition}\label{prop:runtime}
Let $m$ be the number of arcs, B be an upper bound for the capacities and K be the limit on the number of arcs that can be shut in a period. For series-parallel networks \textsc{MFASS} can be solved in time $O(T\log T(KmB)^{2T}T! m)$.
\end{proposition}
\begin{proof}
The first and second component of an entry of a vector in the list at an internal node are bounded by $K$ and $mB$ respectively, hence each entry can take $KmB$ possible values. Therefore every list can contain at most $(KmB)^T$ elements. Thus, the loop over $(z, z') \in L_u \times L_w$ and permutations $\pi$ is over at most $T!(KmB)^{2T}$ elements. If hash tables are used for the check of $z''\in L_v$ then the bound of $O(T\log T)$ for sorting $z''$ dominates the run-time of the loop. In total there are $m-1$ internal nodes, thus the run-time of the complete algorithm is $O(T\log T(KmB)^{2T}T! m)$.
\end{proof}

From Proposition~\ref{prop:runtime}, it follows that for fixed $T$ \textsc{MFASS} on series-parallel networks can be solved in $O(m^{2T+1}B^{2T}K^{2T})$ time where $B$ is the maximum capacity of an arc in the network. Now we use a scaling approach to derive a fully polynomial approximation scheme (FPTAS), that is a family $(\mathcal A_\varepsilon)$ of algorithms, parameterized by a positive real number $\varepsilon$, such that algorithm $\mathcal A_\varepsilon$ produces a solution with objective value at least $(1-\varepsilon)z^*$, where $z^*$ is the optimal value, and the run-time of algorithm $\mathcal A_\varepsilon$ is polynomially bounded in the input size and $1/\varepsilon$. 

Our approximation scheme is based on scaling the problem such that the maximum capacity becomes bounded. In order to ensure that the solution of the scaled problem is sufficiently close to the optimum we need a lower bound for the optimal objective value. If $\lvert J\rvert\leqslant K(T-1)$ there is a feasible solution having one time period without any outage, and the flow value for such a time period will be sufficient as lower bound for our purpose. For $\lvert J\rvert>K(T-1)$ the situation is more complicated, and we need a preprocessing step to transform a given instance into an equivalent one with some control on the maximum capacity. Let $\rho=\max\{0,\lvert J\rvert-K(T-1)\}\in\{0,1,\ldots,K\}$, and let $M$ be the maximum flow value with $\rho$ arcs closed. For $\rho=0$, $M$ is the capacity of a minimum cut and can be computed by solving a max flow problem. For $\rho>0$, the computation of $M$ is described in Algorithm~\ref{alg:get_M}. 
\begin{algorithm}[htb]
  \caption{Computing the maximum flow $M$ with $\rho$ outages}\label{alg:get_M}
  \begin{tabbing}
    .....\=.....\=.....\=.....\=.....\=................... \kill \\
\textbf{for} $v\in\mathcal L$ \textbf{do} \\
\> Let $a\in A$ be the arc corresponding to $v$\\
\> $z_v^0\leftarrow u_a$\\
\> \textbf{if} $a\in J$ \textbf{then} $z_v^1\leftarrow 0$ \textbf{else} $z_v^1\leftarrow-\infty$\\
\> \textbf{for} $j=2,\ldots,\rho$ \textbf{do} $z_v^j\leftarrow -\infty$\\
\textbf{for} $i=d,d-1,\ldots,0$ \textbf{do}\\
\> \textbf{for} $v\in\mathcal W_i$ \textbf{do}\\
\> \> \textbf{for} $j=0,1,\ldots,\rho$ \textbf{do} $z_v^j\leftarrow -\infty$\\
\> \> Let $u$ and $w$ be the child nodes of $v$\\
\> \> \textbf{for} $j=0,1,\ldots,\rho$ \textbf{do}\\
\> \> \> \textbf{for} $j'=0,1,\ldots,\rho-j$ \textbf{do}\\
\> \> \> \> \textbf{if} $v$ is a parallel composition node \textbf{then}\\
\> \> \> \> \> $z_v^{j+j'}\leftarrow\max\{z_v^{j+j'},z_u^j+z_w^{j'}\}$\\
\> \> \> \> \textbf{else}\qquad \{$v$ is a series composition node\}\\
\> \> \> \> \> $z_v^{j+j'}\leftarrow\max\{z_v^{j+j'},\min\{z_u^j,z_w^{j'}\}\}$\\
Let $v$ be the root node\\
return $M=z_v^\rho$
\end{tabbing}
\end{algorithm}
Here, for a node $v$ in the SP-tree and a number $j\in\{0,1,\ldots,\rho\}$, $z_v^j$ is the capacity of the subnetwork corresponding to node $v$ when $j$ arcs in the intersection of $J$ and this subnetwork are closed. If $j$ is larger than the size of this intersection, we put $z_v^j=-\infty$. Algorithm~\ref{alg:get_M} shows that $M$ can be computed efficiently.
\begin{lemma}\label{lem:efficient_M}
The maximum flow value $M$ subject to the constraint that $\rho$ arcs from $J$ carry zero flow can be determined in time $O(mK^2)=O(m^3)$.  \qed
\end{lemma}

No arc can carry more than $M$ units of flow in any time period, hence we may assume w.l.o.g. that $B\leqslant M$. We also know that the optimal objective value is at least $M$ because, we can schedule $\rho$ jobs allowing a flow of value $M$ in time period 1, and then continue arbitrarily. Let $L=\max\{1,\varepsilon B/(mT)\}$ and consider the scaled problem with the capacities $u_a$ replaced by $u'_a=\lfloor u_a/L\rfloor$. The scaled instance can be solved in time
\[O(m^{2T+1}(B/L)^{2T}K^{2T})=O(m^{4T+1}K^{2T} /\varepsilon^{2T}).\]
For any feasible vector $y=(y_{ai})_{a\in A,i\in[T]}\in\{0,1\}^{\lvert J\rvert\, T}$, let $F(y)$ and $F'(y)$ denote the objective values for the problem on the original network  and for the scaled version, respectively. Let $y^*=(y^*_{ai})_{a\in A,i\in[T]}$ and $\tilde y=(\tilde y_{ai})_{a\in A,i\in[T]}$ denote optimal solutions of the problem on the original network and of the scaled version, respectively. In the following lemma, we study the the behaviour of the objective values for these solutions under the scaling.
\begin{lemma}\label{lem:bounds} We have the following estimates:
  \begin{align}
  L\cdot F'(y^*)&\geqslant (1-\varepsilon)F(y^*),\label{eq:bound_1}\\
  F(\tilde y)&\geqslant L\cdot F'(\tilde y).\label{eq:bound_2}
  \end{align}
\end{lemma}
\begin{proof}
Both inequalities are obvious for $L=1$, because in this case the original and the scaled problem coincide. So we assume $L>1$. For $i=1,\ldots,T$ let $C_i$ be a minimum cut in the network $(V,A^*_i,s,t,u')$ where $A^*_i=\{a\in A\,:\,y^*_{ai}=1\}$.
Then, using $B\leqslant M\leqslant F(y^*)$, we obtain
\begin{multline*}
  L\cdot F'(y^*)=L\sum_{i=1}^T\sum_{a\in C_i}u'_{a}\geqslant L\sum_{i=1}^T\left(\sum_{a\in C_i}\frac{u_{a}}{L}-\lvert C_i\rvert\right)\geqslant \sum_{i=1}^T\sum_{a\in C_i}u_{a}-LmT\\
=\sum_{i=1}^T\sum_{a\in C_i}u_{a}-\varepsilon B\geqslant(1-\varepsilon)F(y^*).
\end{multline*}
Similarly, let $C'_i$ be a minimum cut in the network $(V,\tilde A_i,s,t,u)$ where $\tilde A_i=\{a\in A\,:\,\tilde y_{ai}=1\}$.
Then
\[F(\tilde y) = \sum_{i=1}^T\sum_{a\in C'_i}u_{ai}\geqslant L\sum_{i=1}^T\sum_{a\in C'_i}u'_{ai}\geqslant LF'(\tilde y).\qedhere\]
\end{proof}

\begin{proposition}\label{prop:fptas}
For fixed $T$, the class $\mathcal C_{\text{sp}}$ of instances with a series-parallel network has an FPTAS with run-time $O(m^{2T+1}(B/L)^{2T}K^{2T})=O(m^{4T+1}K^{2T} /\varepsilon^{2T})=O(m^{6T+1}/\varepsilon^{2T})$.
\end{proposition}
\begin{proof}
The run-time bound for the scaled problem is a consequence of Proposition~\ref{prop:runtime}, and the approximation guarantee follows from~(\ref{eq:bound_1}) and~(\ref{eq:bound_2}): $F(\tilde y)\geqslant LF'(\tilde y)\geqslant LF'(y^*)\geqslant(1-\varepsilon)F(y^*)$.
\end{proof}

\begin{remark}
The problem can be generalized by allowing the bound on the number of jobs to vary over time. In other words, the parameter $K$ is replaced by a vector $(K_1,\ldots,K_T)$ and constraints~(\ref{eq:job_bound}) are replaced by 
\[\sum_{a\in J} y_{ai} \geqslant \lvert J\rvert - K_i\qquad\text{for all }i\in[T].\]
Algorithm~\ref{alg:sp_networks_K} can be modified to solve this more general problem, and with 
\[\rho=\max\left\{0,\,\lvert J\rvert-\sum_{i=1}^TK_i+\min_{i\in[T]}K_i\right\}\]
we obtain an FPTAS of runtime $O(m^{6T+1}/\varepsilon^{2T})$ for this problem.
\end{remark}

For $K=\lvert J\rvert$, it was shown in~\cite{boland2013unit_time} that the method corresponding to Algorithm~\ref{alg:sp_networks_K} runs in time $O(m^{2T-1}B^{2T-2})$, and using the same argument as above, we obtain the following approximation result. 
\begin{proposition}\label{prop:fptas2}
For fixed $T$, $K=\lvert J\rvert$ and series-parallel networks, \textsc{MFASS} has an FPTAS with run-time $O(m^{2T-1}(B/L)^{2T-2})=O(m^{4T-3} /\varepsilon^{2T-2})$. \qed
\end{proposition}
If $T$ is not fixed we still get a PTAS using the fact that for $K=\lvert J \rvert$ shutting all
arcs in the job set $J$ at the same time gives an approximation ratio of $(1-1/T)$. The basic idea
is that in order to get a $(1-\varepsilon)$-approximation for an instance with arbitrary $T$ we can
distinguish two cases: if $1/T\leqslant\varepsilon$ we schedule all jobs at time 1 and otherwise we
run the $(1-\varepsilon)$-approximation algorithm from Proposition~\ref{prop:fptas2}.  
\begin{corollary}\label{cor:ptas}
For $K=\lvert J\rvert$ and series-parallel networks, \textsc{MFASS} has a PTAS with run-time
\[O\left(f(1/\varepsilon)m^{4/\varepsilon-3}\right)\]
where $\displaystyle f(x)=x^{5x-5/2}e^{x}\log x$.
\end{corollary}
\begin{proof}
Let $\varepsilon>0$ be fixed. If $1/T\leqslant\varepsilon$ we schedule all jobs at time 1. Otherwise $T< 1/\varepsilon$ and we run the $(1-\varepsilon)$-approximation algorithm for $T$. By Proposition~7 in~\cite{boland2013unit_time}, the run-time is bounded by
\begin{multline}\label{eq:runtime_bound}
O\left(T\log(T)T!(mB/L+1)^{2(T-1)}m\right)
=O\left(T\log(T)T!\left(\frac{m^2T}{\varepsilon}+1\right)^{2(T-1)}m\right)\\
=O\left((mT)^{4T-3}\log(T)T!\left(\frac{1}{\varepsilon T}+\frac{1}{(mT)^2}\right)^{2(T-1)}\right).
\end{multline}
We have
\[
\left(\frac{1}{\varepsilon T}+\frac{1}{(mT)^2}\right)^{2(T-1)}
=\left(\frac{m^2T+\varepsilon}{\varepsilon(mT)^2}\right)^{2(T-1)}
=\left(1+\frac{m^2T+\varepsilon-\varepsilon(mT)^2}{\varepsilon(mT)^2}\right)^{2(T-1)}
\]
With $\alpha=m^2T+\varepsilon-\varepsilon(mT)^2$ and $\beta=\varepsilon(mT)^2$ we obtain
\[\left(\frac{1}{\varepsilon T}+\frac{1}{(mT)^2}\right)^{2(T-1)}
= \left[\left(1+\frac{\alpha}{\beta}\right)^{\beta/\alpha}\right]^{2(T-1)\alpha/\beta}\leqslant e^{2(T-1)\alpha/\beta}.\]
Now
\[2(T-1)\frac\alpha\beta
\leqslant 2T\cdot\frac{m^2T+\varepsilon-\varepsilon(mT)^2}{\varepsilon(mT)^2}
=2\cdot\frac{m^2T(1-\varepsilon T)+\varepsilon}{\varepsilon m^2T}
\leqslant 2/\varepsilon+1,\]
and this implies
\[\left(\frac{1}{\varepsilon T}+\frac{1}{(mT)^2}\right)^{2(T-1)} = O(e^{2/\varepsilon}).\]
Substituting into~(\ref{eq:runtime_bound}) yields a run-time bound of 
\[O\left((mT)^{4T-3}\log(T)T!e^{2/\varepsilon}\right),\]
and since all terms are increasing in $T$, we get with $T<1/\varepsilon$ and using Stirling's formula to bound the factorial, that the run-time is bounded by
\begin{multline*}
O\left((1/\varepsilon)^{4/\varepsilon-3}\log(1/\varepsilon)\lceil 1/\varepsilon\rceil!e^{2/\varepsilon}m^{4/\varepsilon-3}\right)
=O\left((1/\varepsilon)^{4/\varepsilon-3}\log(1/\varepsilon)(1/\varepsilon)^{1/\varepsilon}e^{-1/\varepsilon}\sqrt{1/\varepsilon}e^{2/\varepsilon}m^{4/\varepsilon-3}\right) \\
=O\left((1/\varepsilon)^{5/\varepsilon-5/2}\log(1/\varepsilon)e^{1/\varepsilon}m^{4/\varepsilon-3}\right).\qedhere
\end{multline*}
\end{proof}

\subsubsection*{Acknowledgment}
We would like to thank two anonymous referees for valuable comments that significantly improved the
presentation of our results, in particular the proof of Proposition~\ref{prop:K_2_single_node}.


\begin{thebibliography}{10}

\bibitem{Baxter.etal_Incremental_2014}
M.~Baxter, T.~Elgindy, A.~T. Ernst, T.~Kalinowski, and M.~W.~P. Savelsbergh.
\newblock Incremental network design with shortest paths.
\newblock {\em European Journal of Operational Research}, 238(3):675--684,
  November 2014.

\bibitem{boland2013unit_time}
N.~Boland, T.~Kalinowski, R.~Kapoor, and S.~Kaur.
\newblock Scheduling unit time arc shutdowns to maximize network flow over
  time: complexity results.
\newblock {\em Networks}, 63(2):196--202, 2014.

\bibitem{boland2011cap_align-hvcc}
N.~Boland, T.~Kalinowski, H.~Waterer, and L.~Zheng.
\newblock An optimisation approach to maintenance scheduling for capacity
  alignment in the hunter valley coal chain.
\newblock In E.Y. Baafi, R.J. Kininmonth, and I.~Porter, editors, {\em
  Proceedings of the 35th APCOM Symposium: Applications of Computers and
  Operations Research in the Minerals Industry}, pages 887--897. The
  Australasian Institute of Mining and Metallurgy Publication Series, 2011.

\bibitem{boland2012mixed}
N.~Boland, T.~Kalinowski, H.~Waterer, and L.~Zheng.
\newblock Mixed integer programming based maintenance scheduling for the hunter
  valley coal chain.
\newblock {\em Journal of Scheduling}, 16(6):649--659, 2013.

\bibitem{boland2012scheduling}
N.~Boland, T.~Kalinowski, H.~Waterer, and L.~Zheng.
\newblock Scheduling arc maintenance jobs in a network to maximize total flow
  over time.
\newblock {\em Discrete Applied Mathematics}, 163(1):34--52, 2014.

\bibitem{boland2013capalign-hvcc}
N.~Boland, B.~McGowan, A.~Mendes, and F.~Rigterink.
\newblock Modelling the capacity of the hunter valley coal chain to support
  capacity alignment of maintenance activities.
\newblock In J.~Piantadosi, R.S. Anderssen, and J.~Boland, editors, {\em
  MODSIM2013, 20th International Congress on Modelling and Simulation}, pages
  3302--3308. {Modelling and Simulation Society of Australia and New Zealand},
  2013.

\bibitem{boland2011optimizing_hvcc}
N.~Boland and M.~W.~P. Savelsbergh.
\newblock Optimizing the {H}unter {V}alley coal chain.
\newblock In H.~Gurnani, A.~Mehrotra, and S.~Ray, editors, {\em Supply Chain
  Disruptions: Theory and Practice of Managing Risk}. Springer-Verlag London
  Ltd., 2011.

\bibitem{edmonds1965paths}
J.~Edmonds.
\newblock Paths, trees, and flowers.
\newblock {\em Canadian Journal of Mathematics}, 17(3):449--467, 1965.

\bibitem{gabow1990data}
H.~N. Gabow.
\newblock Data structures for weighted matching and nearest common ancestors
  with linking.
\newblock In {\em Proc. 1st ACM-SIAM Symp. on Discrete Algorithms, SODA 1990},
  pages 434--443, 1990.

\bibitem{Kalinowski.etal_Incremental_2015}
T.~Kalinowski, D.~Matsypura, and M.~W.~P. Savelsbergh.
\newblock Incremental network design with maximum flows.
\newblock {\em European Journal of Operational Research}, 242(1):51--62, Apr
  2015.

\bibitem{king1994faster}
V.~King, S.~Rao, and R.~Tarjan.
\newblock A faster deterministic maximum flow algorithm.
\newblock {\em Journal of Algorithms}, 17(3):447--474, 1994.

\bibitem{koch2011flowsovertime}
R.~Koch, E.~Nasrabadi, and M.~Skutella.
\newblock Continuous and discrete flows over time.
\newblock {\em Mathematical Methods of Operations Research}, 73:301--337, 2011.

\bibitem{kotnyek2003annotated}
B.~Kotnyek.
\newblock An annotated overview of dynamic network flows.
\newblock Technical Report 4936, INRIA, 2003.

\bibitem{liden2014survey}
T.~Lid{\'e}n.
\newblock Survey of railway maintenance activities from a planning perspective
  and literature review concerning the use of mathematical algorithms for
  solving such planning and scheduling problems.
\newblock Technical report, Link\"opings universitet, 2014.

\bibitem{nurre2013thesis}
S.~G. Nurre.
\newblock {\em Integrated network design and scheduling problems: Optimization
  algorithms and applications}.
\newblock PhD thesis, Rensselaer Polytechnic Institute, 2013.
\newblock online: \url{http://search.proquest.com/docview/1466024106}.

\bibitem{nurre2012restoring}
S.~G. Nurre, B.~Cavdaroglu, J.~E. Mitchell, T.~C. Sharkey, and W.~A. Wallace.
\newblock Restoring infrastructure systems: An integrated network design and
  scheduling ({INDS}) problem.
\newblock {\em European Journal of Operational Research}, 223(3):794--806,
  2012.

\bibitem{nurre2014integrated}
S.~G. Nurre and T.~C. Sharkey.
\newblock Integrated network design and scheduling problems with parallel
  identical machines: Complexity results and dispatching rules.
\newblock {\em Networks}, 2014.

\bibitem{orlin2013max}
J.~B. Orlin.
\newblock Max flows in {$O(nm)$} time, or better.
\newblock In {\em Proc. 45th ACM Symp. on Theory of Computing (STOC 2013)},
  pages 765--774. ACM, 2013.

\bibitem{pinedo2008scheduling}
M.~Pinedo.
\newblock {\em Scheduling: Theory, Algorithms, and Systems}.
\newblock Springer, 2008.

\bibitem{skutella2009introduction}
M.~Skutella.
\newblock An introduction to network flows over time.
\newblock In W.~Cook, L.~Lovasz, and J.~Vygen, editors, {\em Research Trends in
  Combinatorial Optimization}, pages 451--482. Springer, 2009.

\bibitem{Tawa}
M.~Tawarmalani and Y.~Li.
\newblock Multi-period maintenance scheduling of tree networks with minimum
  flow disruption.
\newblock {\em Naval Research Logistics}, 58(5):507--530, 2011.

\bibitem{williamson2011design}
D.~P. Williamson and D.~B. Shmoys.
\newblock {\em The design of approximation algorithms}.
\newblock Cambridge University Press, 2011.

\end{thebibliography}


\end{document}